\documentclass{aims}
\usepackage{amsmath}
\usepackage{bm}
  \usepackage{paralist}
  \usepackage{graphics} 
  \usepackage{epsfig} 
\usepackage{graphicx}  \usepackage{epstopdf}
 \usepackage[colorlinks=true]{hyperref}

\usepackage[colorlinks=true]{hyperref}
 \hypersetup{urlcolor=blue, citecolor=red}
\usepackage{pslatex}
\usepackage{amsfonts,color,morefloats}
\usepackage{amssymb,amsmath,latexsym,amsthm}

\def\currentvolume{X}
 \def\currentissue{X}
  \def\currentyear{2018}
   \def\currentmonth{XX}
    \def\ppages{X--XX}
     \def\DOI{10.3934/xx.xx.xx.xx}

\newcommand{\Div}{\operatorname{div}}

\newcommand{\ceil}[1]{{\left\lceil {#1} \right\rceil}}
\newcommand{\floor}[1]{{\left\lfloor {#1} \right\rfloor}}

\newcommand{\res}{\operatorname{res}}
\newcommand{\supp}{\operatorname{supp}}

\newtheorem{theorem}{Theorem}[section]
\newtheorem{corollary}{Corollary}[section]

\newtheorem{lemma}[theorem]{Lemma}
\newtheorem{proposition}{Proposition}

\theoremstyle{definition}
\newtheorem{definition}[theorem]{Definition}

\newtheorem{example}{Example}

\allowdisplaybreaks

\title[Multi-point Codes from the GGS Curves]
      {Multi-point Codes from the GGS Curves}

\author[Chuangqiang Hu, Shudi Yang]{}


\email{huchq@mail2.sysu.edu.cn}
\email{yangshd3@mail2.sysu.edu.cn}

\thanks{This work is partially supported by the NSFC (11701317, 11531007, 11571380, 11701320, 61472457) and Tsinghua University startup fund. This work is also partially supported by China Postdoctoral Science Foundation Funded Project (2017M611801), Jiangsu Planned Projects for Postdoctoral Research Funds (1701104C), Guangzhou Science and Technology Program (201607010144) and the Natural Science Foundation of Shandong Province of China (ZR2016AM04). }
\thanks{$^*$Corresponding author: Shudi Yang.}
\begin{document}
\maketitle

\centerline{\scshape Chuangqiang Hu }
\medskip
{\footnotesize
  \centerline{Yau Mathematical Sciences Center, Tsinghua University,}
    \centerline{Peking, 100084, China}
}

\medskip

\centerline{\scshape Shudi Yang$^*$}
\medskip
{\footnotesize
  \centerline{School of Mathematical Sciences, Qufu Normal University}
    \centerline{Shandong, 273165, China }
}

\bigskip

\centerline{(Communicated by ***)}

\begin{abstract}
This paper is concerned with the construction of algebraic-geometric (AG) codes defined from GGS curves. It is of significant use to describe bases for the Riemann-Roch spaces associated with {some rational places}, 
which enables us to study multi-point AG codes. 
Along this line, we characterize explicitly the Weierstrass semigroups and pure gaps by an exhaustive computation {for the basis of} Riemann-Roch spaces from GGS curves. In addition, we determine the floor of a certain type of divisor and investigate the properties of AG codes. Multi-point codes with excellent parameters are found, among which, a presented code with parameters $ [216,190,\geqslant 18] $ over $ \mathbb{F}_{64} $ yields a new record.
\end{abstract}

\section{Introduction}\label{sec:intro}

In the early 1980s, Goppa \cite{goppa1977codes} constructed algebraic geometric  codes (AG codes for short) from algebraic curves. Since then, the study of AG codes becomes an important instrument in the theory of error-correcting codes.  
Roughly speaking, the parameters of an AG code are good
when the underlying curve has many rational points with respect
to its genus. For this reason maximal curves, that is curves attaining the Hasse-Weil upper bound, have been widely investigated in the literature: for example the Hermitian curve and its quotients, the Suzuki curve, the Klein quartic and the GK curve. In this work we will study multi-point AG codes on the GGS curves. 

In order to construct good AG codes we need to study Weierstrass semigroups and pure gaps. Their use dates back to the theory of one-point codes. For example, the authors in \cite{Guruswami,stichtenoth,Yang,Yang2} examined one-point codes from Hermitian curves and developed efficient methods to decode them.  Korchm\'{a}ros and Nagy \cite{Korchmaros} computed the Weierstrass semigroup of a degree-three closed point of the Hermitian curve. 
Matthews \cite{matthews2005weierstrass} determined the Weierstrass semigroup of
any $ r $-tuple rational points on the quotient of a Hermitian curve. As is known, Weierstrass pure gap is also a useful tool in coding theory.   
Garcia, Kim and Lax improved the Goppa bound {using Weierstrass gaps} at one place in \cite{garcia1993consecutive,garcia1992goppa}.
The concept of pure gaps of a pair of points on a curve was initiated by Homma and Kim \cite{Homma2001Goppa}, and it had been pushed forward by Carvalho and Torres \cite{carvalho2005goppa} to several points. 
Maharaj and Matthews \cite{maharaj2006floor} extended this construction by introducing the notion of the floor of a divisor and obtained improved bounds on the parameters of AG codes. The authors in \cite{Bartoli2016,matthews2001} counted the gaps and pure gaps at two points on Hermitian curves and those from some specific Kummer extensions, respectively. Recently, Yang and Hu \cite{Yanghu1,Yanghu2} extended their work by investigating Weierstrass semigroups and pure gaps at many points. 

We mention that Maharaj \cite{maharaj2004code} showed that Riemann-Roch spaces of divisors from fiber
products of Kummer covers of the projective line, can be decomposed as a direct sum of Riemann-Roch
spaces of divisors of the projective line. Maharaj, Matthews and Pirsic \cite{Maharaj2005riemann} determined explicit bases for large classes of Riemann-Roch spaces
of the Hermitian function field. Along this research line, Hu and Yang \cite{HuYang2016Kummer} gave other explicit bases for Riemann-Roch spaces of divisors over Kummer extensions, which makes it convenient to determine the pure gaps.

In this work, we focus our attention on the GGS curves, which are maximal curves constructed by Garcia,  G{\"u}neri and Stichtenoth \cite{garcia2010GGS} over $ \mathbb{F}_{q^{2n}} $ defined by the equations
\begin{align*} 
\left\{\begin{array}{ll}
x^{q \phantom{^2}}+x&=y^{q+1},\\
y^{q^2}-y&=z^m,
\end{array}	
\right.
\end{align*}
where $ q $ is a prime power and $ m=(q^n+1)/(q+1) $ with $ n >1 $ to be an odd integer. Obviously the GGS curve is a generalization of the GK curve initiated by Giulietti and Korchm{\'a}ros \cite{Giulietti2008} where we take $ n=3 $. Recall that Fanali and Giulietti \cite{Fanali2010GK} exhibited one-point AG codes on the GK curves and obtained linear codes with better parameters with respect to those { previously known}. Two-point and multi-point AG codes on the GK maximal curves were studied in \cite{Bartoli2017GK} and \cite{Castellanos2016GK}, respectively. Bartoli, Montanucci and Zini \cite{Bartoli2017} examined one-point AG codes from the GGS curves.
Inspired by the above work and \cite{Masuda2,HuYang2016Kummer}, here we will investigate multi-point AG codes arising from GGS curves. To be precise, an explicit basis for the corresponding Riemann-Roch space is determined by constructing a related set of lattice points. The properties of AG codes from GGS curves are also considered.  Then the basis is utilized to  characterize the Weierstrass semigroups and pure gaps with respect to several rational places. In addition, we give an effective algorithm to compute the floor of divisors. Finally, our results will lead us to find new codes with better parameters in comparison with the existing codes in  MinT's Tables \cite{MinT}. A new record-giving $ [216,190,\geqslant 18] $-code over $ \mathbb{F}_{64} $ is presented as one of the examples.

The remainder of the paper is organized as follows. 
Section \ref{sec:Bases} focuses on the construction of bases for the Riemann-Roch space from GGS curves. Section \ref{sec:proAG} studies the properties of the related AG codes. In Section \ref{sec:Weiersemipureg} we illustrate the Weierstrass semigroups and the pure gaps. Section \ref{sec:floor} devotes to the floor of divisors from GGS curves.  Finally, in Section \ref{sec:Examples} we employ our results to construct multi-point codes with excellent parameters.

\section{Bases for Riemann-Roch spaces { of the GGS curves}}\label{sec:Bases}

Throughout this paper, we always let $q$ be a prime power and $ n >1 $ be an odd integer. The GGS curve $ \textup{GGS}(q,n) $ over $ \mathbb{F}_{q^{2n}} $ is defined by the equations
\begin{align}\label{eq:GGScurve}
\left\{\begin{array}{ll}
x^{q \phantom{^2}}+x&=y^{q+1},\\
y^{q^2}-y&=z^m,
\end{array}	
\right.
\end{align}
where $ m=(q^n+1)/(q+1) $. 
The genus of $ \textup{GGS}(q,n) $ is $  (q-1)(q^{n+1}+q^n-q^2)/2 $ and there are $ q^{2n+2} -q^{n+3 }+ q^{n+2} + 1 $ rational points, 
see \cite{garcia2010GGS} for more details. Especially when $ n=3 $, it becomes the well-known maximal curve introduced by Giulietti and Korchm\'{a}ros \cite{Giulietti2008}, the so-called GK curve, which is not a subcover of the corresponding Hermitian curve.

{Let $ \alpha,\beta,\gamma \in \mathbb{F}_{q^{2n}} $ such that $ \alpha^q+\alpha=\beta^{q+1}  $ and $ \beta^{q^2} -\beta = \gamma^m $. In particular, if $ \gamma = 0  $, then $ \alpha$, $\beta \in \mathbb{F}_{q^{2}} $. 
Denote by $ \mathcal{P}_{\alpha,\beta,\gamma} $ the rational place of the function field $  \mathbb{F}_{q^{2n}}( \textup{GGS}(q,n)) $ centered at the rational affine point with coordinates $(\alpha,\beta,\gamma) $ and $ P_{\infty} $ the place located at infinity. They are exactly all the $  \mathbb{F}_{q^{2n}} $-rational places}. 
Take $ Q_\beta:= \sum_{\alpha^q+\alpha=\beta^{q+1}}\mathcal{P}_{\alpha,\beta,0}  $ where $ \beta \in \mathbb{F}_{q^2} $ and we always view $ \mathbb{F}_{q^2}  $ as a subfield of $ \mathbb{F}_{q^{2n}} $. Notice that $ \deg (Q_\beta)=q $. {For later use, we write $ P_\mu := \mathcal{P}_{\alpha_\mu,0,0} $ where $ \alpha_{\mu} \in \mathbb{F}_{q^{2}} $ with $ \alpha_\mu^q+\alpha_\mu=0 $ and $ 0\leqslant \mu <q $}. Particularly we denote $ P_0:=\mathcal{P}_{0,0,0} $ and
$ Q_0:=P_0+P_1+\cdots+P_{q-1} $.

The following proposition describes some principle divisors {of} the GGS curve.
\begin{proposition}\label{prop:divisor}
	Let the curve $ \textup{GGS}(q,n) $ be given in \eqref{eq:GGScurve} and assume that {$ \alpha_{\mu} \in \mathbb{F}_{q^{2}} $} with $ 0 \leqslant \mu <q $ are the solutions of $ x^q+x=0 $. Then we obtain
	\begin{enumerate}
		\item [$ (1) $]	$ \Div (x-\alpha_{\mu})=m(q+1)P_{\mu}-m(q+1)P_{\infty}$,
		\item [$ (2) $] $ \Div  (y-\beta)=m Q_{\beta} -mq P_{\infty} $  for  $ \beta \in \mathbb{F}_{q^2} $,
		\item [$ (3) $] $ \Div (z)=\sum_{\beta \in \mathbb{F}_{q^2}} Q_{\beta}-q^3P_{\infty}$.
	\end{enumerate}		
\end{proposition}

For convenience, we use $ Q_{ \nu} $ ($ 0 \leqslant \nu \leqslant q^2-1 $) to represent the divisors $ Q_{\beta} $  $( \beta \in \mathbb{F}_{q^2}) $. 
{In particular, the symbol $ Q_0 $ performs well in both cases $ v = 0 \in  \mathbb{Z}  $ and $ \beta=0 \in \mathbb{F}_{q^2} $. }

For a function field $ F $, the Riemann-Roch vector space with respect to a divisor $ G $ is defined by
\begin{align*}
\mathcal{L}(G)=\Big\{f \in F\,\,\Big| \,\, \Div(f)+G \geqslant 0 \Big\} \cup \{0\}.
\end{align*} 
Let $ \ell(G) $ be the dimension of $ \mathcal{L}(G) $. From the famous {Riemann-Roch Theorem \cite[Theorem 1.5.15]{stichtenoth}}, we know that
\begin{align*}
\ell(G)-\ell(W-G)=1-g+\deg (G),
\end{align*}
where $ W $ is a canonical divisor and $ g $ is the genus of the associated curve. 

{In this section, we consider divisors supported at all rational places of the form $ G:=\sum_{\mu=0}^{q-1} r_{\mu}P_{\mu} + \sum_{\nu=1}^{q^2-1}s_{\nu} Q_{\nu} +tP_{\infty} $, where $ r_{\mu} $'s, $ s_{\nu}$'s and $ t $ are integers.} We wish to show that the Riemann-Roch space $ \mathcal{L}(G) $ is generated by some elements, say $ E_{i,\,\bm{j},\,\bm{k}} $ {for some $ i,\bm{j},\bm{k} $}, and the number of such elements equals $ \ell(G) $. For this purpose, we proceed as follows.

Let $ {\bm{j}}=(j_1,j_2,\cdots,j_{q-1}) $ and $ {\bm{k}}=(k_1,k_2,\cdots,k_{q^2-1}) $. 
For $ (i,{\bm{j}},{\bm{k}})  \in \mathbb{Z}^{q^2+q-1} $, we define
\begin{align}\label{eq:Eijk}
E_{i,\,\bm{j},\,\bm{k}} := z^i \prod_{\mu=1}^{q-1}(x-\alpha_{\mu})^{j_{\mu}} \prod_{\nu=1}^{q^2-1}(y- \beta_{\nu})^{k_{\nu}}  .
\end{align}
{ Here and thereafter, we denote $  |{\bm{v}}| $ to be the sum of all the coordinates of a given vector $ {\bm{v}} $. Then $ |{\bm{j}}|=\sum_{\mu=1}^{q-1} j_{\mu} $ and  $ |{\bm{k}}|=\sum_{\nu=1}^{q^2-1} k_{\nu} $.} By Proposition \ref{prop:divisor}, one can compute the divisor of $ E_{i,\,\bm{j},\,\bm{k} } $:
\begin{align}\label{eq:E}
\Div (E_{i,\,\bm{j},\,\bm{k}} )=&iP_0 + \sum_{\mu=1}^{q-1}(i+m(q+1) j_{\mu}) P_{\mu}
+\sum_{\nu=1}^{q^2-1}(i+m k_{\nu})Q_{\nu} \nonumber \\
& -\big(q^3 i+m(q+1)|{\bm{j}}|+mq|{\bm{k}}| \big) P_{\infty} . 
\end{align}

For later use, we denote by $ \mathbb{N}_0 := \mathbb{N}\cup \{0\} $ the set of nonnegative integers
and denote by $ \lceil x \rceil $ the smallest integer not less than $ x $. It is easy to show
that $ j = \left\lceil  {a}/{b} \right\rceil $ if and only if
$ 0 \leqslant b j -a < b $, where $ b \in \mathbb{N} $ and $ a\in \mathbb{Z} $.

Put $ {\bm{r}}=(r_0,r_1,\cdots,r_{q-1}) $ and $ {\bm{s}}=(s_1,s_2,\cdots,s_{q^2-1}) $.
Let us define a set of lattice points for $ (\bm{r},\,\bm{s},\,t)\in \mathbb{Z}^{q^2+q} $,
\begin{align*}
\Omega_{\bm{r},\,\bm{s},\,t}  := \Big\{& (i,\bm{j},\,\bm{k})
\,\,\Big| \,\, i+r_0 \geqslant 0, \\
& 0 \leqslant i + m(q+1)j_{\mu} + r_{\mu} < m(q+1) \text{ for }  \mu =1,\cdots, q-1,\\
& 0 \leqslant i+m k_{\nu}+s_{\nu} <m \text{ for }  \nu =1,\cdots, q^2-1,\\
&q^3 i+m(q+1)|{\bm{j}}|+mq|{\bm{k}}|  \leqslant t
\Big\},
\end{align*}
or equivalently,
\begin{align}
\Omega_{\bm{r},\,\bm{s},\,t}  := \Big\{& (i,\bm{j},\,\bm{k})
\,\,\Big| \,\,i+r_0 \geqslant 0, \nonumber \\
& j_{\mu} =  \ceil{\frac{-i-r_{\mu}}{m(q+1)}  }  \text{ for }   \mu =1,\cdots, q-1,\nonumber \\
& k_{\nu} =  \ceil{\frac{-i-s_{\nu}}{m}  }  \text{ for }   \nu =1,\cdots, q^2-1,\nonumber \\
&q^3 i+m(q+1)|{\bm{j}}|+mq|{\bm{k}}|    \leqslant t
\Big\}.\label{eq:omegarceiljmu}
\end{align}

Our key result depends on the following lemma. However, {its proof is technical and it will be completed later.} 

\begin{lemma}\label{thm:omega}
	{Assume that $ t\geqslant 2g-1 -q^3 w $, where $ w $ is the minimum element of all the coordinates of $ \bm{r} $ and $ \bm{s} $, then the cardinality of $ \Omega_{\bm{r},\,\bm{s},\,t} $ is }
	\begin{equation*}
	\#\Omega_{\bm{r},\,\bm{s},\,t} = 1-g+t+|{\bm{r}}|+q|{\bm{s} }|.
	\end{equation*}
\end{lemma}

 Now we can easily prove the main result of this section.
\begin{theorem}\label{thm:basis1}
	Let $ G:=\sum_{\mu=0}^{q-1} r_{\mu}P_{\mu} + \sum_{\nu=1}^{q^2-1}s_{\nu} Q_{\nu} +tP_{\infty} $. The elements $E_{i,\,\bm{j},\,\bm{k}} $ with $ (i,\bm{j},\,\bm{k}) \in \Omega_{\bm{r},\,\bm{s},\,t}  $ constitute a basis for the Riemann-Roch space
	$ \mathcal{L}(G) $. In particular
	$ \ell(G) = \# \Omega_{\bm{r},\,\bm{s},\,t} $.
\end{theorem}
\begin{proof}
	Let  $ (i,\bm{j},\,\bm{k}) \in \Omega_{\bm{r},\,\bm{s},\,t}  $. It follows from the definition that $ E_{i,\,\bm{j},\,\bm{k}} \in \mathcal{L}(G) $, where $ G=\sum_{\mu=0}^{q-1} r_{\mu}P_{\mu} + \sum_{\nu=1}^{q^2-1}s_{\nu} Q_{\nu} +tP_{\infty}  $.
	From Equation \eqref{eq:E}, we have
	$ v_{P_0}(E_{i,\,\bm{j},\,\bm{k}}) = i$, which indicates that the valuation of $E_{i,\,\bm{j},\,\bm{k}}$ at the rational place $P_0$ uniquely depends on $i$. Since lattice points in $ \Omega_{\bm{r},\,\bm{s},\,t}  $ provide distinct values of $i$, the elements $ E_{i,\,\bm{j},\,\bm{k}} $
	are linearly independent of each other, with $ (i,\bm{j},\,\bm{k}) \in \Omega_{\bm{r},\,\bm{s},\,t}  $. In order to indicate that they constitute a basis for $ \mathcal{L}(G) $, the only thing left is to prove that 
	\[
	\ell (G)
	= \#\Omega_{\bm{r},\,\bm{s},\,t} .
	\]	
	For the case of $ r_0 $ sufficiently large, it follows from the Riemann-Roch Theorem and Lemma \ref{thm:omega} that
	\begin{align*}
	\ell (G) & = 1-g + \deg(G)\\
	& = 1-g +|{\bm{r}}|+q|{\bm{s} }|+t= \#\Omega_{\bm{r},\,\bm{s},\,t} .
	\end{align*}
	This implies that
	$ \mathcal{L}(G) $
	is spanned by elements $ E_{i,\,\bm{j},\,\bm{k}} $ with $ (i,\bm{j},\,\bm{k})$ in the set $  \Omega_{\bm{r},\,\bm{s},\,t}  $.
	
	For the general case, we choose $ r_0'> r_0 $ large enough and set 
	$ G' :=r_0' P_0+\sum_{\mu=1}^{q-1} r_{\mu}P_{\mu} + \sum_{\nu=1}^{q^2-1}s_{\nu} Q_{\nu} +tP_{\infty}  $ and $ {\bm{r'}}=(r_0',r_1,\cdots,r_{q-1}) $.
	From above argument, we know that the elements $ E_{i,\,\bm{j},\,\bm{k}} $ with $ (i,\bm{j},\,\bm{k}) \in \Omega_{{\bm{r',s}},t} $ span the whole space of	$ \mathcal{L}(G') $.
	Remember that
	$ \mathcal{L}(G) $ is a linear subspace of $ \mathcal{L}(G') $, which can be written as
	\begin{equation*}
	\mathcal{L}(G) = \Big\{ f \in \mathcal{L}(G') \,\,\Big| \,\,v_{P_0}(f)\geqslant -r_0\Big\}.
	\end{equation*}
	Thus, we choose $ f \in \mathcal{L}(G)  $ and suppose that
	\begin{equation*}
	f=\sum_{(i,\bm{j},\,\bm{k}) \in \Omega_{\bm{r'},\,\bm{s},\,t}} a_{i} E_{i,\,\bm{j},\,\bm{k}},
	\end{equation*}
	since $ f \in \mathcal{L}(G') $ by definition. The valuation of $f$ at $ P_0 $ is $ v_{P_0}(f)=\min_{a_i\neq 0} \{ i \}$. Then the inequality $ v_{P_0}(f)\geqslant -r_0 $ gives that, if $ a_i \neq 0 $, then $ i \geqslant -r_0 $. Equivalently, if $ i < -r_0 $, then $ a_i=0 $.  From the definition of $ \Omega_{\bm{r},\,\bm{s},\,t}  $ and $ \Omega_{{\bm{r',s}},t} $, we get that
	\begin{equation*}
	f=\sum_{(i,\bm{j},\,\bm{k}) \in \Omega_{\bm{r},\,\bm{s},\,t}  } a_{i} E_{i,\,\bm{j},\,\bm{k}}.
	\end{equation*}
	Then the theorem follows.
\end{proof}

We now turn to prove Lemma \ref{thm:omega} which requires a series of results including telescopic semigroups listed as follows.

\begin{definition}[{\cite{Kirfel1995}, Definition 6.1}]
	Let $ (a_1,\cdots, a_k) $ be a sequence of positive integers such that the greatest common divisor is $ 1 $. Define $ d_i = \gcd (a_1,\cdots, a_i)  $ and \[ A_i = \{ a_1/d_i,\cdots, a_i/d_i\} \] for $ i =
	1,\cdots, k $. Let $  d_0 = 0 $. Let $ S_i $ be the semigroup generated by $ A_i $. If $ a_i/d_i \in S_{i-1} $ for
	$ i = 2,\cdots, k $, we call the sequence $ (a_1,\cdots, a_k ) $ telescopic. A semigroup is called telescopic if it is generated by a telescopic sequence.	
\end{definition}

\begin{lemma}[{\cite{Kirfel1995}, Lemma 6.4}]\label{lem:teli}
	If $ (a_1,\cdots, a_k) $ is telescopic and $  M \in S_k $, then there exist uniquely determined non-negative integers $ 0 \leqslant x_i <d_{i-1}/d_i $ for $ i = 2,\cdots, k $, such that
	\begin{align*}
	M = \sum_{i=1}^{k} x_i a_i.
	\end{align*}
	We call this representation the normal representation of $ M $.
\end{lemma} 

\begin{lemma}[{\cite{Kirfel1995}, Lemma 6.5}]\label{lem:genus} 
	For the semigroup generated by the telescopic sequence $ (a_1,\cdots, a_k) $ we have
	\begin{align*}
	l_g(S_k) &= \sum_{i=1}^{k} (d_{i-1}/d_i - 1)a_i,\\
	g(S_k) & =  (l_g(S_k) + 1)/2,
	\end{align*}
	where $ l_g(S_k) $ and $ g(S_k) $ denote the largest gap and the
	number of gaps of $ S_k $, respectively.  
\end{lemma}

\begin{lemma}\label{lem:PsiR}
	Let $ m=(q^n+1)/(q+1) $, $ g= (q-1)(q^{n+1}+q^n-q^2)/2$ for an odd integer $ n >1 $. Let $ t \in \mathbb{Z} $. Consider the lattice point set $ \Psi(t) $ defined by
	\begin{align*}
	\Big\{ (a,b,c)
	\,\,\Big| \,\,0 \leqslant a <m,\,\, 0 \leqslant b \leqslant q,\,\, c \geqslant 0, \,\,  
	q^3 a+mq b+m(q+1) c  \leqslant t\Big\},
	\end{align*}
	If $ t \geqslant 2g-1 $, then  $ \Psi(t) $ has cardinality 
	\begin{equation*}
	\#\Psi(t)= 1-g + t.
	\end{equation*}
\end{lemma}
\begin{proof}
	Let $ a_1=q^3$, $a_2=mq $, $a_3=m(q+1) $. It is easily verified that the sequence $ (a_1,a_2,a_3) $ is telescopic. By Lemma \ref{lem:teli} every element $ M $ in $ S_3 $ has a unique representation $ M=a_1 a+a_2 b+a_3 c $, where $ S_3 $ is the semigroup generated by $ (a_1,a_2,a_3) $.
	One obtains from Lemma \ref{lem:genus} that
	\begin{align*}
	l_g(S_3) &= (q-1)(q^n+1)(q+1)-q^3,\\
	g(S_3) & =  \frac{1}{2}(l_g(S_3) + 1) =\frac{1}{2}(q-1)(q^{n+1}+q^n-q^2)=g.
	\end{align*}
	It follows that the set $ \Psi(t) $ has cardinality 
	$ 1-g + t $ provided that $ t \geqslant 2g-1 =   l_g(S_3) $, which finishes the proof.	
\end{proof}
From Lemma \ref{lem:PsiR}, we get the number of lattice points in $ \Omega_{\bm{0},\,\bm{0},\, t} $.
\begin{lemma}\label{lem:omega0000t}
	If $ t\geqslant 2g-1 $, then the cardinality of $ \Omega_{\bm{0},\,\bm{0},\, t} $ is
	\begin{equation*}
	\#\Omega_{\bm{0},\,\bm{0},\,t} = 1-g+t.
	\end{equation*}
\end{lemma}

\begin{proof}
	Note that
	\begin{align}\label{eq:omega000t}
	\Omega_{\bm{0},\,\bm{0},\,t} :=  \Big\{& (i,\bm{j},\,\bm{k})
	\,\,\Big| \,\,i  \geqslant 0, \nonumber \\
	& j_{\mu} =  \ceil{\frac{-i }{m(q+1)}  }  \text{ for }   \mu =1,\cdots, q-1,\nonumber \\
	& k_{\nu} =  \ceil{\dfrac{-i }{m}  }  \text{ for }   \nu =1,\cdots, q^2-1,\nonumber \\
	&q^3 i+m(q+1)|{\bm{j}}|+mq|{\bm{k}}|    \leqslant t
	\Big\}. 
	\end{align}
	Set $ i := a+ m(b+(q+1)c)  $ with $ 0\leqslant a < m $, $ 0\leqslant b \leqslant q $ and $ c \geqslant 0 $.  
	Then Equation \eqref{eq:omega000t} gives that   
	\begin{align*}
	\Omega_{\bm{0},\,\bm{0},\, t} \cong  
	\Big\{ (a,b,c)
	\,\,\Big| \,\,0 \leqslant a <m, \,\,0 \leqslant b \leqslant q,\,\, c \geqslant 0,  \,\, 
	q^3 a+mq b+(q^n+1) c  \leqslant t\Big\}.
	\end{align*} 
	Here and thereafter,
	the notation 
	$ A \cong B $ means that two lattice point sets $ A $ and $ B $ are bijective. {In the last formula, the bijection comes from the mappings $ (i,j,k ) \mapsto ( a+ m(b+(q+1)c),-c,-b-(q+1)c)  $ and $ (a,b,c )\mapsto ( i+mk,(q+1)j-k,-j)  $ by denoting $ j=j_1=\cdots=j_{q-1} $ and $ k=k_1=\cdots=k_{q^2-1} $.}
	 Thus the assertion $ \#\Omega_{\bm{0},\,\bm{0},\,t} = 1-g+t $ is derived from Lemma \ref{lem:PsiR}. 
\end{proof}

The next three lemmas state some elementary properties of $ \Omega_{\bm{r},\,\bm{s},\,t}  $.
\begin{lemma}\label{lem:omeganoorder}
	The lattice point set $ \Omega_{\bm{r},\,\bm{s},\,t}  $ as defined above is symmetric with respect to $ r_0,r_1,\cdots,r_{q-1}  $ and $ s_1,s_2,\cdots,s_{q^2-1}  $, respectively. In other words, we have $ \# \Omega_{\bm{r},\,\bm{s},\,t}  = \# \Omega_{\bm{r'},\,\bm{s'},\,t}  $ by denoting $ \bm{r'}=(r_0',r_1',\cdots,r_{q-1}')  $ and $ \bm{s'}=(s_1',s_2',\cdots,s_{q^2-1}')  $, where the sequences  $ \big(r_i\big)_{i=0}^{q-1}   $ and $ \big(s_i\big)_{i=1}^{q^2-1}   $ are equal to $ \big(r_i'\big)_{i=0}^{q-1}    $ and $ \big(s_i'\big)_{i=1}^{q^2-1}   $ up to permutation, respectively.
\end{lemma}
\begin{proof}
	Recall that $ \Omega_{\bm{r},\,\bm{s},\,t}  $  is defined by
	\begin{align*}
	\Omega_{\bm{r},\,\bm{s},\,t}  = \Big\{& (i',\bm{j'},\bm{k'})
	\,\,\Big| \,\,i'+r_0 \geqslant 0, \nonumber \\
	& j_{\mu}' =  \ceil{\frac{-i'-r_{\mu}}{m(q+1)}  }  \text{ for }   \mu =1,\cdots, q-1,\nonumber \\
	& k_{\nu}' =  \ceil{\frac{-i'-s_{\nu} }{m}  }  \text{ for }   \nu =1,\cdots, q^2-1,\nonumber \\
	&q^3 i'+m(q+1)|{\bm{j'}}|+mq|{\bm{k'}}|    \leqslant t
	\Big\},
	\end{align*}
	where $ {\bm{j'}}=(j_1',\cdots,j_{q-1}') $ and $ {\bm{k'}}=(k_1',\cdots,k_{q^2-1}') $. 
	It is important to write $ i'=i+m(q+1)l $ with $ 0 \leqslant i < m(q+1) $. Let $  j_{\mu}' =  j_{\mu}-l $ for $  \mu \geqslant 1 $, $  k_{\nu}' =  k_{\nu}-(q+1)l $ for $  \nu \geqslant 1 $. Then
	\begin{align*}
	\Omega_{\bm{r},\,\bm{s},\,t}  \cong\Big\{& (i,l,\bm{j},\,\bm{k})
	\,\,\Big| \,\,i+m(q+1)l \geqslant -r_0, \,\,
	0 \leqslant i <  m(q+1), \nonumber \\
	& j_{\mu} =  \ceil{ \frac{-i-r_{\mu}}{m(q+1)} } \text{ for }  \mu =1,\cdots, q-1 ,\nonumber \\
	& k_{\nu} =  \ceil{\frac{-i-s_{\nu} }{m}  }  \text{ for }   \nu =1,\cdots, q^2-1,\nonumber \\
	&q^3 i+m(q+1)\big(l+|{\bm{j}}|\big)+mq|{\bm{k}}|    \leqslant t
	\Big\},
	\end{align*} 
	where $ {\bm{j}}=(j_1,\cdots,j_{q-1}) $ and $ {\bm{k}}=(k_1,\cdots,k_{q^2-1}) $.
	The first inequality in $ \Omega_{\bm{r},\,\bm{s},\,t}  $ gives that  $ l\geqslant j_{0}: = \ceil{ \dfrac{-i-r_{0}}{m(q+1)} }  $. So we write $ l = j_{0} +\iota $ with $ \iota \geqslant 0 $.
	Then
	\begin{align*}
	\Omega_{\bm{r},\,\bm{s},\,t} \cong \Big\{& (i,\iota, j_0,\bm{j},\,\bm{k})
	\,\,\Big| \,\,
	0 \leqslant i < m(q+1), \,\,\iota \geqslant 0,\\
	& j_{\mu} =  \ceil{ \frac{-i-r_{\mu}}{m(q+1)} } \text{ for }  \mu =0,1,\cdots, q-1 ,\nonumber \\
	& k_{\nu} =  \ceil{\frac{-i-s_{\nu} }{m}  }  \text{ for }   \nu =1,\cdots, q^2-1,\nonumber \\
	&q^3 i+m(q+1)\big(j_0+\iota+|{\bm{j}}|\big)+mq|{\bm{k}}|    \leqslant t
	\Big\}.
	\end{align*}
	The right hand side means that the number of the lattice points does not depend on the order of $r_{\mu} $, $ 0 \leqslant \mu \leqslant q-1 $, and the order of $s_{\nu}$, $ 1 \leqslant \nu \leqslant q^2-1 $, which concludes the desired assertion.		
\end{proof}

\begin{lemma}\label{lem:omegas2=0}
	Let $ {\bm{r}}=(r_0,r_1,\cdots,r_{q-1}) \in \mathbb{N}_0^q$ and $ {\bm{s}}=(s_1,s_2,\cdots,s_{q^2-1})\in \mathbb{N}_0^{q^2-1} $. If $ t \geqslant 2g-1 $, then
	\[
	\# \Omega_{\bm{r},\,\bm{s} ,\, t} = \# \Omega_{\bm{0},\,\bm{s},\,t} + |{\bm{r}}| .
	\]
\end{lemma}
\begin{proof} 
	Let us take the sets $ \Omega_{\bm{r},\,\bm{s},\,t}  $ and $ \Omega_{\bm{r'},\,\bm{s},\,t} $ for consideration, where $ {\bm{r'}}=(0,r_1,\cdots,r_{q-1}) $. It follows from the definition that the complement set   $\Delta:=\Omega_{\bm{r},\,\bm{s},\,t} \backslash \Omega_{\bm{r'},\,\bm{s},\,t} $ is given by
	\begin{align*}
	\Big\{& (i,\bm{j},\,\bm{k})
	\,\,\Big| \,\, -r_0 \leqslant i < 0, \\
	& j_{\mu} =  \ceil{ \frac{-i-r_{\mu}}{m(q+1)} } \text{ for }  \mu =1,\cdots, q-1 ,\nonumber \\
	& k_{\nu} =  \ceil{\frac{-i-s_{\nu} }{m}  }  \text{ for }   \nu =1,\cdots, q^2-1,\nonumber \\
	&q^3 i+m(q+1)|{\bm{j}}|+mq|{\bm{k}}|  \leqslant t
	\Big\}.
	\end{align*}
	{Clearly} $ \Delta=\varnothing $ if $ r_0=0 $. To determine the cardinality of $ \Delta $ with $ r_0>0 $, we denote $ i := a+ m(b+(q+1)c)  $ with integers $ a,b,c $ satisfying $0\leqslant  a <m $, $ 0\leqslant b \leqslant q $ and $ c \leqslant -1 $.  
	Then $ j_{\mu} \leqslant -c $ for $ \mu \geqslant 1 $,
	$ k_{\nu} \leqslant -b-(q+1)c  $ for $ \nu \geqslant 1 $. A straightforward computation shows
	\begin{align*} 
	q^3 i+m(q+1)|{\bm{j}}|+mq|{\bm{k}}|
	&  \leqslant 
	q^3 a+ mqb +m(q+1)c \\
	& \leqslant q^3 (m-1)+ mq^2 -m(q+1)=2g-1.
	\end{align*}
	So the last inequality in $ \Delta $ always	
	holds for all $ t\geqslant 2g-1 $, which means that the cardinality of $ \Delta $ is determined by the first inequality, that is $ \#\Delta = r_0 $. Then we must have
	\begin{align*}
	\# \Omega_{\bm{r},\,\bm{s},\,t} =\# \Omega_{\bm{r'},\,\bm{s},\,t}   +r_0,
	\end{align*}
	whenever $ r_0 \geqslant 0 $. 
	Repeating the above {argument} and using Lemma \ref{lem:omeganoorder}, we get
	\begin{align*}
	\# \Omega_{\bm{r},\,\bm{s},\,t}= \# \Omega_{\bm{0},\,\bm{s},\,t} + |{\bm{r}}| , 
	\end{align*}
	where $ {\bm{r}}=(r_0,r_1,\cdots,r_{q-1}) $.	
\end{proof}

\begin{lemma}\label{lem:omegasr=0}
	Let  $ {\bm{s}}=(s_1,s_2,\cdots,s_{q^2-1})\in \mathbb{N}_0^{q^2-1} $. If $ t \geqslant 2g-1 $, then the following identity holds:
	\[
	\# \Omega_{\bm{0},\,\bm{s},\, t} =   \# \Omega_{\bm{0},\,\bm{0},\, t}+q |{\bm{s}}| .
	\]
\end{lemma}
\begin{proof}
	For convenience, let us denote $ {\bm{r}}:=(s_0,s_0,\cdots,s_0) $ to be the $ q $-tuple with all entries equal $ s_0 $, where $  s_0 \geqslant 0  $, and write $ \Omega_{\bm{r},\,\bm{s},\,t}  $ as $ \Gamma_{s_0,(s_1,\cdots,\,s_{q^2-1}),\,t} $. 
	To get the desired conclusion, we first claim that 
	\begin{align}\label{eq:theta}
	\# \Gamma_{s_0,(s_1,\cdots,s_{q^2-1}),\,t}=\#\Gamma_{s_0',(s_1',\cdots,s_{q^2-1}'),\,t},
	\end{align}
	where the sequence  $ \big(s_i\big)_{i=0}^{q^2-1}  $ is equal to $ \big(s_i'\big)_{i=0}^{q^2-1} $ up to permutation.
	
	Note that $ \Gamma_{s_0,(s_1,\cdots,s_{q^2-1}),\,t} $ is equivalent to 
	\begin{align*}
	\Big\{& (i',j',k_1',\cdots,k_{q^2-1}')
	\,\,\Big| \,\, i'+s_0 \geqslant 0, \\
	& 0 \leqslant i' + m(q+1)j'+s_0  < m(q+1),\\
	& 0 \leqslant i'+m k_{\nu}'+s_{\nu} <m \text{ for }  \nu =1,\cdots, q^2-1,\\
	&q^3 i'+m(q+1)(q-1)j'+mq|{\bm{k}'}|  \leqslant t
	\Big\}, 
	\end{align*}
	where $ |{\bm{k}'}|=\sum_{\nu=1}^{q^2-1} k_{\nu}'  $. By setting $ i':=i+m\kappa $ where $ 0\leqslant i  <m $ and $ k_{\nu}':=k_{\nu} -\kappa $ for $ \nu \geqslant 1 $, we obtain
	\begin{align*} \Gamma_{s_0,(s_1,\cdots,s_{q^2-1}),\,t} \cong 
	\Big\{& (i,\kappa,j',k_1,\cdots,k_{q^2-1})
	\,\,\Big| \,\,0\leqslant i  <m,\,\, i+m\kappa+s_0 \geqslant 0,  \\
	& 0 \leqslant i+m\kappa + m(q+1)j'+s_0  < m(q+1),\\
	& 0 \leqslant i+m k_{\nu}+s_{\nu} <m \text{ for }  \nu =1,\cdots, q^2-1,\\
	&q^3 i+m(q+1)(q-1)j'+mq(\kappa+|{\bm{k}}|)  \leqslant t
	\Big\}. 
	\end{align*}
	Put $ \kappa :=k_0+\varepsilon $ where $ k_0 =\ceil{\dfrac{-i-s_0}{m}} $ and $ \varepsilon =-(q+1)j_0+\eta $ with $ 0\leqslant \eta < q+1 $. One gets that $ 0 \leqslant i+m k_{0}+s_{0} <m $, which leads to $ 0 \leqslant i+m \kappa +m(q+1)j_0-m\eta+s_{0} <m $. So the inequality $ 0 \leqslant i+m \kappa +m(q+1)j_0+s_{0} <m+m\eta \leqslant m(q+1) $ holds because $ \varepsilon =-(q+1)j_0+\eta $. Thus we must have $ j_0=j'\leqslant 0 $. Therefore
	\begin{align*} \Gamma_{s_0,(s_1,\cdots,s_{q^2-1}),\,t} \cong 
	\Big\{& (i, j_0,\eta,k_0, k_1,\cdots,k_{q^2-1})
	\,\,\Big| \,\,0\leqslant i  <m,\,\, j_0 \leqslant 0,  \\
	& 0\leqslant \eta < q+1,\\
	&  k_{\nu}=\ceil{\dfrac{-i-s_{\nu}}{m}} 
	\text{ for }  \nu =0,1,\cdots, q^2-1,\\
	&q^3 i- m(q+1)j_0+mq(k_0+\eta+|{\bm{k}}|)  \leqslant t
	\Big\}. 
	\end{align*}
	The right hand side means that the lattice points do not depend on the order of $s_{\nu} $ with $ 0 \leqslant \nu \leqslant q^2-1 $, by observing that $   k_{\nu} $ is determined by $ s_{\nu} $.  In other words, we have shown that the number of lattice points in $ \Gamma_{s_0,(s_1,\cdots,s_{q^2-1}),\,t} $ does not depend on the order of $ s_{\nu} $ with $ 0\leqslant \nu \leqslant q^2-1 $, concluding 
	the claim we presented by \eqref{eq:theta}. So it follows from \eqref{eq:theta} and Lemma \ref{lem:omegas2=0} that	
	\begin{align*}
	\# \Gamma_{0,\,(s_1,\,s_2,\,\cdots,\,s_{q^2-1}),\,t}
	&=\#\Gamma_{s_1,\,(0,\,s_2,\,\cdots,\,s_{q^2-1}),\,t}\\
	&=\#\Gamma_{0,\,(0,\,s_2,\,\cdots,\,s_{q^2-1}),\,t} +	qs_1.
	\end{align*}
	By repeatedly using Lemma \ref{lem:omegas2=0}, we get
	\begin{align*}
	\# \Gamma_{0,\,(s_1,\,s_2,\,\cdots,\,s_{q^2-1}),\,t}
	=\#\Gamma_{0,\,(0,\,0,\,\cdots,0),\,t} +	q (s_1+s_2+\cdots+s_{q^2-1}),
	\end{align*}
	concluding the desired formula $ \# \Omega_{\bm{0},\,\bm{s},\, t} =   \# \Omega_{\bm{0},\,\bm{0},\, t}+q |{\bm{s}}| $.	
\end{proof}

We are now in a position to give the proof of Lemma \ref{thm:omega}.
\begin{proof}{[\bf{Proof of Lemma \ref{thm:omega}}]}
	By taking 
	$ w:=\min_{0 \leqslant \mu \leqslant q-1 \atop 1 \leqslant \nu \leqslant q^2-1 }\Big\{r_{\mu},s_{\nu} \Big\} $, we obtain from the definition that
	$ \Omega_{\bm{r},\,\bm{s},\,t}  $ is equivalent to $ \Omega_{\bm{r'},\,\bm{s'},\,t'} $, where $ {\bm{r'}}=(r_0-w,\cdots,r_{q-1}-w) $, $ {\bm{s'}}= (s_1-w,\cdots,s_{q^2-1}-w)$ and $ t'=t+q^3 w $. Hence $ \# \Omega_{\bm{r},\,\bm{s},\,t} =\# \Omega_{\bm{r'},\,\bm{s'},\,\,t'} $.  
	On the other hand, by observing that $ r_{\mu}-w \geqslant 0 $, $ s_{\nu}-w \geqslant 0 $ and $ t'\geqslant 2g-1 $, we establish from Lemmas \ref{lem:omega0000t}, \ref{lem:omegas2=0} and \ref{lem:omegasr=0} that
	\begin{align*}
	\# \Omega_{\bm{r'},\,\bm{s'},\,t'}
	&=\# \Omega_{\bm{0},\,\bm{s'},\,t'} + |{\bm{r'}}| \\
	&=\# \Omega_{\bm{0},\,\bm{0},\,t'} + q  |{\bm{s'}}| + |{\bm{r'}}| \\
	&=1-g+t'+ q  |{\bm{s'}}| + |{\bm{r'}}|\\
	&=1-g+t+ q  |{\bm{s}}| + |{\bm{r}}| .	
	\end{align*}
	It then follows that 
	\begin{align*} 
	\# \Omega_{\bm{r},\,\bm{s},\,t} =1-g+t+ q  |{\bm{s}}| + |{\bm{r}}|, 	 
	\end{align*} completing the proof of Lemma \ref{thm:omega}.	
\end{proof}

We finish this section with a result that allows us to give a new form of the basis for our Riemann-Roch space $ \mathcal{L}(G) $ with $ G=\sum_{\mu=0}^{q-1} r_{\mu}P_{\mu} + \sum_{\nu=1}^{q^2-1}s_{\nu} Q_{\nu} +tP_{\infty} $. Denote $ {\bm{\lambda}}:=(\lambda_1,\cdots,\lambda_{q-1}) $ and $ {\bm{\gamma}}:=(\gamma_1,\cdots,\gamma_{q^2-1})  $. For $ (u,\bm{\lambda},\bm{\gamma}) \in \mathbb{Z}^{q^2+q-1} $, we define 
\begin{equation*}
\Lambda_{u,\bm{\lambda},\bm{\gamma}}:=\tau ^u \prod_{\mu=1}^{q-1} f_{\mu}^{\lambda_{\mu}} \prod_{\mu=1}^{q^2-1} h_{\nu}^{\gamma_{\nu}} ,
\end{equation*}
where $ \tau:=\dfrac{z^{q^{n-3}}}{x-\alpha_0 }  $, $ f_{\mu}: = \dfrac{x-\alpha_{\mu}}{x-\alpha_0} $ for $ \mu \geqslant 1 $, and $  h_{\nu}:=  \dfrac{y-\beta_{\nu}}{y-\beta_0} $ for $ \nu \geqslant 1 $. {We have from Proposition \ref{prop:divisor} that}
\begin{align*}
\Div(\Lambda_{u,\bm{\lambda},\bm{\gamma}} )=
& \sum_{\mu=1}^{q-1}\Big(q^{n-3}u+m(q+1) \lambda_{\mu}-m |{\bm{\gamma}}|\Big) P_{\mu} 
+\sum_{\nu=1}^{q^2-1}(q^{n-3}u+m \gamma_{\nu})Q_{\nu} \nonumber \\
& -\Big((m(q+1)-q^{n-3}) u+m(q+1)|{\bm{\lambda}}| + m |{\bm{\gamma}}|\Big)P_0+u P_{\infty} .
\end{align*}

There is a close relationship between the elements
$ \Lambda_{u,\bm{\lambda},\bm{\gamma} } $ and $ E_{i,\bm{j},\,\bm{k} } $ explored as follows.
\begin{corollary}\label{cor:thetabase2}
	Let  $ G:=\sum_{\mu=0}^{q-1} r_{\mu}P_{\mu} + \sum_{\nu=1}^{q^2-1}s_{\nu} Q_{\nu} +tP_{\infty} $. Then the elements $ \Lambda_{u,\bm{\lambda},\bm{\gamma}} $ with $ (u,\bm{\lambda},\bm{\gamma})\in \Theta_{\bm{r},\,\bm{s},\,t} $ form a basis for the Riemann-Roch space
	$ \mathcal{L}(G) $, where the set $ \Theta_{\bm{r},\,\bm{s},\,t} $ is given by
	\begin{align}\label{eq:thetaaaa}
	\Big\{& (u,\bm{\lambda},\bm{\gamma}) \,\,\Big| \,\,  u   \geqslant - t,\nonumber \\
	& 0 \leqslant q^{n-3}u+m \gamma_{\nu}+s_{\nu} <m \textup{ for }  \nu =1,\cdots, q^2-1 \nonumber\\
	& 0 \leqslant q^{n-3}u+m(q+1) \lambda_{\mu}-m |{\bm{\gamma}}|+  r_{\mu} < m(q+1) \textup{ for }  \mu =1,\cdots, q-1, \nonumber \\
	& (m(q+1)-q^{n-3}) u+m(q+1)|{\bm{\lambda}}| + m |{\bm{\gamma}}| \leqslant  r_0  
	\Big\}.
	\end{align} 
	In addition we have $ \#\Theta_{\bm{r},\,\bm{s},\,t}=\# \Omega_{\bm{r},\,\bm{s},\,t}   $.
\end{corollary}
\begin{proof}
	It suffices to prove that the set
	\begin{align*}
	&\Big\{ \Lambda_{u,\bm{\lambda},\bm{\gamma}} \,\,\Big| \,\,(u,\bm{\lambda},\bm{\gamma})\in \Theta_{\bm{r},\,\bm{s},\,t} \Big\} \\
	\intertext{equals the set}	
	&\Big\{ E_{i,\,\bm{j},\,\bm{k}} \,\,\Big| \,\, (i,\bm{j},\,\bm{k})\in \Omega_{\bm{r},\,\bm{s},\,t} \Big\}.
	\end{align*}
	In fact, for fixed $ (u,\bm{\lambda},\bm{\gamma}) \in \mathbb{Z}^{q^2+q-1}$, we obtain $\Lambda_{u,\bm{\lambda},\bm{\gamma}}$ equals $ E_{i,\,\bm{j},\,\bm{k}} $ with
	\begin{align*}
	i&= - (m(q+1)-q^{n-3}) u-m(q+1)|{\bm{\lambda}}| - m |{\bm{\gamma}}| , \\
	j_{\mu}&=u+|{\bm{\lambda}}|+\lambda_{\mu} \text{ for }  \mu =1,\cdots, q-1, \\
	k_{\nu}&= (q+1)(u+|{\bm{\lambda}}|) +|{\bm{\gamma}}|+\gamma_{\nu}  \text{ for }  \nu =1,\cdots, q^2-1.
	\end{align*}
	On the contrary, if we set
	\begin{align*}
	u&=-q^3 i-m(q+1)|{\bm{j}}|-mq|{\bm{k}}| ,\\
	\lambda_{\mu}&=q^2 i +q^{n-1}  |{\bm{j}}|+m|{\bm{k}}|+j_{\mu} \text{ for }  \mu =1,\cdots, q-1, \\ 
	\gamma_{\nu} &= (q+1)(i+q^{n-3}|{\bm{j}}|)+q^{n-2}|{\bm{k}}|+k_{\nu} \text{ for }  \nu =1,\cdots, q^2-1,
	\end{align*}
	then $ E_{i,\,\bm{j},\,\bm{k}} $ is exactly the element $ \Lambda_{u,\bm{\lambda},\bm{\gamma}} $.
	Therefore, if we restrict $ (i,\bm{j},\,\bm{k}) $ in $ \Omega_{\bm{r},\,\bm{s},\,t}  $, then we must have  $(u,\bm{\lambda},\bm{\gamma})$ is in $ \Theta_{\bm{r},\,\bm{s},\,t} $ and vice versa. This completes the proof of this corollary. 
\end{proof}

In the following, we will demonstrate an interesting property of $ \# \Omega_{{\bm{r}} ,{\bm{s}},t} $ for GK curves with a specific vector $ {\bm{s}} $.  
\begin{corollary}\label{cor:sym_t_r0}
	Let $ n=3 $ and the vectors $ {\bm{r,s}} $ be given by
	\begin{align*}
	{\bm{r}}&:=(r_0,r_1,\cdots,r_{q-1}),\\
	{\bm{s}}&:=  ( s_1', s_2', \cdots,s_{q^2-1}' ) \\
	&\phantom{:}=(
	\underbrace{s_1,s_1,\cdots,s_1}_{q+1}, \underbrace{s_2,s_2,\cdots,s_2}_{q+1},\cdots,
	\underbrace{s_{q-1},s_{q-1},\cdots,s_{q-1}}_{q+1}) . 
	\end{align*}  
	Then the lattice point set $ \Omega_{\bm{r},\,\bm{s},\,t}  $ is symmetric with respect to $ r_0,r_1,\cdots,r_{q-1},t  $. In other words, we have	$ \# \Omega_{\bm{r},\,\bm{s},\,t}  = \# \Omega_{\bm{r'},\,\bm{s},\,t'}  $, where the sequence $ \big(r_i\big)_{i=0}^{q}   $ is equal to $ \big(r_i'\big)_{i=0}^{q} $ up to permutation by putting $   r_q :=t $ and $   r_q' :=t' $. 	
\end{corollary}
\begin{proof}	
	Denote $ {\bm{r}}:=(r_0,\dot{{\bm{r}}} )=(r_0,r_1,\cdots,r_{q-1})  $ and $ \Omega_{(r_0,\dot{{\bm{r}}} ),{\bm{s}},\,t} :=\Omega_{\bm{r},\,\bm{s},\,t}  $. By   Lemma \ref{lem:omeganoorder}, it suffices to prove that 
	\begin{align*} 
	\# \Omega_{(r_0,\dot{\bm{r}} ),\,\bm{s},\,t}
	=\#\Theta_{(r_0,\dot{\bm{r}} ),\bm{s},\,t}
	=\#\Omega_{(t,\dot{\bm{r}} ),\,\bm{s},\,r_0}.
	\end{align*}

	The first identity follows directly from Corollary \ref{cor:thetabase2}.
	Applying Corollary \ref{cor:thetabase2} again gives the set $ \Theta_{(r_0,\dot{\bm{r}} ),\bm{s},\,t} $ as
	\begin{align*} 
	\Big\{& (u,\bm{\lambda},\bm{\gamma}) \,\,\Big| \,\,  u + t  \geqslant 0,\nonumber \\
	& 0 \leqslant u+m \gamma_{\nu}+s_{\nu}' <m \text{ for }  \nu =1,\cdots, q^2-1   , \nonumber \\
	& 0 \leqslant u+m(q+1) \lambda_{\mu}-m |{\bm{\gamma}}|+  r_{\mu} < m(q+1) \text{ for }  \mu =1,\cdots, q-1, \nonumber \\
	&  q^3 u+m(q+1)|{\bm{\lambda}}| + m |{\bm{\gamma}}|  \leqslant r_0 
	\Big\}.
	\end{align*} 
	From our assumption, it is obvious that $ |{\bm{\gamma}}| $ is divisible by $ q+1 $. So if we take $ i:=u $, $ j_{\mu}:= \lambda_{\mu}-\dfrac{|{\bm{\gamma}}|}{q+1} $ for $ \mu \geqslant 1 $ and $ k_{\nu} := \gamma_{\nu}$ for $ \nu \geqslant 1 $, 
	then $ \Theta_{(r_0,\dot{{\bm{r}}} ),{\bm{s}},\,t} $ is equivalent to
	\begin{align*} 
	\Big\{& (i,\bm{j},\,\bm{k}) \,\,\Big| \,\,  i + t  \geqslant 0,\nonumber \\
	& 0 \leqslant i+m k_{\nu}+s_{\nu}' <m \text{ for }  \nu =1,\cdots, q^2-1 , \nonumber \\
	& 0 \leqslant i+m(q+1) j_{\mu}+  r_{\mu} < m(q+1) \text{ for }  \mu =1,\cdots, q-1, \nonumber \\
	& { q^3 i+m(q+1)|{\bm{j}}| + m q|{\bm{k}}|  \leqslant r_0 }  
	\Big\}.
	\end{align*} 
	The last set is exactly $ \Omega_{(t,\dot{{\bm{r}}} ),{\bm{s}},r_0} $ by definition. Hence the second identity is just shown, completing the whole proof. 	 	
\end{proof}

\section{The AG codes from GGS curves}\label{sec:proAG}

This section settles the properties of AG codes from GGS curves. Generally speaking, there are two classical ways of constructing AG codes associated with divisors $ D $ and $ G $, where $ G $ is a divisor of arbitrary function field $ F $ and $ D:=Q_1+\cdots+Q_N $ is another divisor of $ F $ such that $ Q_1,\cdots,Q_N $ are pairwise distinct rational places, each not belonging to the support of $ G $. One construction is based on the Riemann-Roch space $ \mathcal{L}(G) $,
\begin{equation*}
C_{\mathcal{L}}(D,G):=\Big\{ (f(Q_1),\cdots,f(Q_N))\,\,\Big| \,\,f\in \mathcal{L}(G) \Big\} \subseteq \mathbb{F}_q^N.
\end{equation*}
The other one depends on the space of differentials $ \Omega(G-D) $,
\begin{equation*}
C_{\Omega}(D,G):= \Big\{ (\res_{Q_1}(\eta),\cdots,\res_{Q_N}(\eta))\,\,\Big| \,\, \eta \in \Omega(G-D)  \Big\}.
\end{equation*}
It is well-known the codes $ C_{\mathcal{L}}(D,G) $ and $ C_{\Omega}(D,G) $ are dual to each other. Further $ C_{\Omega}(D,G) $ has parameters $ [N,k_{\Omega},d_{\Omega}] $ with 
$ k_{\Omega} =N-k $ and $ d_{\Omega}\geqslant \deg(G)-(2g-2) $, where $ k= \ell( G)-\ell(G-D)$ is the dimension of $ C_{\mathcal{L}}(D,G) $. If moreover $ 2g-2 < \deg(G) < N $ then \begin{equation*} 
k_{\Omega}=N+g-1-\deg(G).
\end{equation*}
The reader is referred to \cite{stichtenoth} for more information.

{In this section, we follow the notation given in Section \ref{sec:Bases}.
	Let $  D  $ be the direct sum of all $ \mathbb{F}_{q^{2n}} $-rational places except the $ \mathbb{F}_{q^{2}} $-rational places of the function field $ \mathbb{F}_{q^{2n}}(\textup{GGS}(q, n) ) $, namely,  $ D:=\sum_{\alpha,\beta,\gamma \atop \gamma \neq 0} \mathcal{P}_{\alpha,\beta,\gamma} $, where $ \alpha,\beta,\gamma $ are elements in $\mathbb{F}_{q^{2n}} $ satisfying $ \alpha^q+\alpha=\beta^{q+1}  $ and $ \beta^{q^2} -\beta = \gamma^m $.
	Now we will study the AG code
$ C_{\mathcal{L}}(D,G)  $ with $ G:=\sum_{\mu=0}^{q-1} r_{\mu}P_{\mu} + \sum_{\nu=1}^{q^2-1}s_{\nu} Q_{\nu}+tP_{\infty} $. }
The length of $ C_{\mathcal{L}}(D,G)  $ is 
\begin{align*}
N:=\deg(D)= q^{n+2}(q^n-q+1)-q^3.
\end{align*}
It is well known that the dimension of $  C_{\mathcal{L}}(D,G) $ is given by
\begin{align}\label{eq:dim}
\dim C_{\mathcal{L}}(D,G)  =\ell(G)-\ell(G-D) .
\end{align}  
Set $ R:=N+2g-2 $. {If} $ \deg(G) > R $, we deduce from the Riemann-Roch Theorem and \eqref{eq:dim} that 
\begin{align*}
\dim C_{\mathcal{L}}(D,G) 
& = (1-g+\deg (G))- (1-g+\deg(G-D))\\
& = \deg (G)-\deg(G-D)=N,
\end{align*} 
which implies that $ C_{\mathcal{L}}(D,G) $ is trivial. So we only consider the case $ 0\leqslant \deg(G) \leqslant R $.

Now, we use the following lemmas to calculate the dual of $ C_{\mathcal{L}}(D,G) $. 
\begin{lemma}[\cite{stichtenoth}, Proposition 2.2.10]\label{lem:dual}
	Let $ \tau $ be an element {of a function field} such that $ v_{P_i}(\tau)=1 $ for all rational places $ P_i $ contained in
	the divisor $ D $. Then the dual of $ C_{\mathcal{L}}(D,G)$ is
	\begin{align*}
	C_{\mathcal{L}}(D,G)^{\bot} = C_{\mathcal{L}} (D,D-G+\Div(d\tau)-\Div(\tau)).
	\end{align*}
\end{lemma}

\begin{lemma}[\cite{stichtenoth}, Proposition 2.2.14]\label{lem:eqdual}
	Suppose that $ G_1 $ and $ G_2 $ are divisors with $ G_1=G_2+\Div(\rho) $ for some $ \rho\in F\backslash \{0\} $ and $ \supp G_1 \cap \supp D = \supp G_2 \cap \supp D =\varnothing $. Let $ N:=\deg (D) $ and $ \bm{\varrho}:= ( \rho(P_1),\cdots,\rho(P_N))   $ with
	$ P_i \in D $. Then the codes 
	$ C_{\mathcal{L}}(D,G_1)$ and $ C_{\mathcal{L}}(D,G_2)$ are equivalent and	\begin{align*}
	C_{\mathcal{L}}(D,G_2)  = \bm{\varrho} \cdot C_{\mathcal{L}} (D,G_1).
	\end{align*}
\end{lemma}

The dimension and the dual code of $ C_{\mathcal{L}}(D,G) $ can be described below.
\begin{theorem}\label{th:dualcode}
	Let $ A:= (q^n+1)(q-1)-1 $, $ B:=mq^2(q^n-q^3)+(q^n+1)(q^2-1)-1  $  and $ \rho:= 1+\sum_{i=1}^{\frac{n-3}{2}} z^{(q^n+1)(q-1)   \sum_{j=1}^{i} q^{2j} } $. Then the dual code of $ C_{\mathcal{L}}(D,G) $ is given as follows.
	\begin{enumerate}	 
		\item [$ (1) $] The dual of $ C_{\mathcal{L}}(D,G) $ is represented as
		\begin{align*}
		C_{\mathcal{L}}(D,G)^{\bot} = \bm{\varrho} \cdot C_{\mathcal{L}} (D,\sum_{\mu=0}^{q-1} (A-r_{\mu})P_{\mu} + \sum_{\nu=1}^{q^2-1}(A-s_{\nu}) Q_{\nu}+(B-t)P_{\infty}),
		\end{align*}
		where $ \bm{\varrho}:= ( \rho(\mathcal{P}_{\alpha_1,\beta_1,\gamma_1}),\cdots,\rho(\mathcal{P}_{\alpha_N,\beta_N,\gamma_N}) )   $ with
		$ \mathcal{P}_{\alpha_i,\beta_i,\gamma_i} \in D $. 
		\item [$ (2) $] In particular, for $ n=3 $, we have $ \rho=1 $ and 
		\begin{align*}
		C_{\mathcal{L}}(D,G)^{\bot} = C_{\mathcal{L}} (D,\sum_{\mu=0}^{q-1} (A-r_{\mu})P_{\mu} + \sum_{\nu=1}^{q^2-1}(A-s_{\nu}) Q_{\nu}+(B-t)P_{\infty}).
		\end{align*}
	\end{enumerate}
\end{theorem}
\begin{proof}
	Define 
	\begin{align*}
	H:=\Big\{z\in \mathbb{F}_{q^{2n}}^* \,\,\Big| \,\, \exists y \in \mathbb{F}_{q^{2n}} \textup{ with } y^{q^2}-y=z^m \Big\}. 
	\end{align*}	
	Consider the element
	\begin{align*}
	\tau:=\prod_{\gamma\in H }(z-\gamma).
	\end{align*}
	Then $ \tau $ is a prime element for all places $ \mathcal{P}_{\alpha,\beta,\gamma} $ in $ D $ and its divisor is
	\begin{align*}
	\Div(\tau)=\sum_{\gamma\in H }\Div(z-\gamma)=D-\deg(D) P_{\infty} ,
	\end{align*}	
	where  $ D=\sum_{\alpha,\beta,\gamma \atop \gamma \neq 0}\mathcal{P}_{\alpha,\beta,\gamma} $ and $ N =\deg(D)=q^{n+2}(q^n-q+1)-q^3 $.
	Moreover by a same discussion as in the proof of Lemma 2 in \cite{abdon2009further}, we have
	\begin{align*}
	\tau= 1+\sum_{i=0}^{k-1}w^{  \sum_{j=0}^i q^{2j} + \sum_{j=0}^{k-1} q^{2j+1} } + \sum_{i=0}^{k-1}w^{ \sum_{j=0}^i   q^{2j+1} }, 
	\end{align*}
	where $ n=2k+1 $ (note that $ n>1 $ is odd) and $ w=z^{(q^n+1)(q-1)} $. Then a straightforward computation shows
	\begin{align*}
	d\tau&= w^{   \sum_{j=0}^{k-1} q^{2j+1} } 
	\Big( 1+\sum_{i=1}^{k-1} w^{   \sum_{j=1}^{i} q^{2j} }\Big) dw\\
	&= w^{   \frac{q^n-q}{q^2-1} } 
	\Big( 1+\sum_{i=1}^{k-1} w^{   \sum_{j=1}^{i} q^{2j} }\Big) dw,\\
	dw&=-z^{(q^n+1)(q-1)-1}dz.
	\end{align*}
	Let $ \rho:= 1+\sum_{i=1}^{k-1} w^{   \sum_{j=1}^{i} q^{2j} } $ and denote its divisor by $ \Div(\rho) $. Set 
	\begin{align*}
	A&:= m(q^n-q)+(q^n+1)(q-1)-1 ,\\
	S&:=q^3A - 2g+2.
	\end{align*}
	Since $ \Div(dz)=(2g-2)P_{\infty} $ (see Lemma 3.8 of \cite{Guneri2013}), it follows from Proposition \ref{prop:divisor} that  
	\begin{align*}
	\Div(d\tau)&=A\cdot \Div(z)+ \Div(dz)+\Div(\rho)\\
	&=A \sum_{\beta \in \mathbb{F}_{q^2}} Q_{\beta} -\Big(q^3A - 2g+2\Big)P_{\infty}+\Div(\rho)\\
	& =A \sum_{\mu=0}^{q-1} P_{\mu} + A \sum_{\nu=1}^{q^2-1}  Q_{\nu}-S P_{\infty}+\Div(\rho).
	\end{align*}
	Let $ \eta:= d\tau/\tau $ be a Weil differential. The divisor of $ \eta $ is
	\begin{align*}
	\Div(\eta)&=\Div(d\tau)-\Div(\tau)\\
	& =A \sum_{\mu=0}^{q-1} P_{\mu} + A \sum_{\nu=1}^{q^2-1}  Q_{\nu}  -D +\
	\Big(\deg(D)-S\Big) P_{\infty} +\Div(\rho) .
	\end{align*}	 
	By writing $ B:=\deg(D)-S=mq^2(q^n-q^3)+(q^n+1)(q^2-1)-1 $, we establish from Lemma \ref{lem:dual} that the dual of $ C_{\mathcal{L}}(D,G) $ is  
	\begin{align*}
	C_{\mathcal{L}}(D,G)^{\bot} 
	&= C_{\mathcal{L}} (D,D-G+\Div(\eta))\\
	&=C_{\mathcal{L}} (D,\sum_{\mu=0}^{q-1} (A-r_{\mu})P_{\mu} + \sum_{\nu=1}^{q^2-1}(A-s_{\nu}) Q_{\nu}+(B-t)P_{\infty} +\Div(\rho)).
	\end{align*} 
	Denote $ \bm{\bm{\varrho}}:= ( \rho(\mathcal{P}_{\alpha_1,\beta_1,\gamma_1}),\cdots,\rho(\mathcal{P}_{\alpha_N,\beta_N,\gamma_N}) )   $ with
	$ \mathcal{P}_{\alpha_i,\beta_i,\gamma_i} \in D $. Then we deduce the first statement from Lemma \ref{lem:eqdual}. The second statement then follows immediately.
\end{proof}

\begin{theorem}\label{th:dim}
	Suppose that $ 0\leqslant \deg(G) \leqslant R $. Then the dimension of $ C_{\mathcal{L}}(D,G) $ is given by
	\begin{align*}
	\dim C_{\mathcal{L}}(D,G)=\left\{\begin{array}{lll} 
	\# \Omega_{\bm{r},\,\bm{s},\,t}    && \textup{ if } 0\leqslant \deg(G) < N,\\
	N-\# \Omega_{\bm{r},\,\bm{s},\,t} ^{\bot}  && \textup{ if } N \leqslant \deg(G) \leqslant R,
	\end{array}
	\right.
	\end{align*}
	where $ \Omega_{\bm{r},\,\bm{s},\,t} ^{\bot}:= \Omega_{\bm{r'},\,\bm{s'},\,B-t}  $  with
	$ {\bm{r'}}=(A-r_0,\,\cdots,\,A-r_{q-1}) $ and $ {\bm{s'}}=(A-s_1,\,\cdots,\,A-s_{q^2-1}) $.

\end{theorem}
\begin{proof}
	For $ 0\leqslant \deg(G) < N $, we have by Theorem \ref{thm:basis1} and Equation \eqref{eq:dim} that
	\begin{align*}
	\dim  C_{\mathcal{L}}(D,G) = \ell(G) =  \# \Omega_{\bm{r},\,\bm{s},\,t} . 
	\end{align*} 
	
	For $  N \leqslant \deg(G) \leqslant R $, Theorem \ref{th:dualcode} yields that
	\begin{align*}
	\dim  C_{\mathcal{L}}(D,G) = N-\dim  C_{\mathcal{L}}(D,G)^{\bot} = N- \# \Omega_{\bm{r},\,\bm{s},\,t} ^{\bot}. 
	\end{align*} 
	So the proof is completed. 
\end{proof}

{Goppa bound provides an estimate for the minimum distance of $ C_{\mathcal{L}}(D,G) $.} Techniques for improving the Goppa bounds will be dealt with in the next two sections.

\section{Weierstrass semigroups and pure gaps}\label{sec:Weiersemipureg}
In this section, we will characterize the Weierstrass semigroups and pure gaps {on GGS curves}, which will enables us to obtain improved bounds on the parameters of AG codes. 

We need some preliminary notation and results before we begin. For an arbitrary function field $ F $, let $ Q_1,\cdots, Q_l $ be distinct rational places of $ F $, then the Weierstrass semigroup
$ H(Q_1,\cdots, Q_l) $ is defined by
\[
\Big\{(s_1,\cdots, s_l)\in \mathbb{N}_0^l\,\,\Big| \,\,\exists f\in F \text{ with } (f)_{\infty}=\sum_{i=1}^l s_i Q_i  \Big\},
\]
and the Weierstrass gap set $  G(Q_1,\cdots, Q_l)  $ is defined by $ \mathbb{N}_0^l \backslash H(Q_1,\cdots, Q_l) $, where $ \mathbb{N}_0 $ denotes the set of nonnegative integers. The details are found in  \cite{matthews2004weierstrass}.

Homma and Kim \cite{Homma2001Goppa} introduced the concept of pure gap set with respect to a pair of rational places.  This was generalized by
Carvalho and Torres \cite{carvalho2005goppa} to several rational places, denoted by $ G_0(Q_1,\cdots, Q_l) $, which is given by
\begin{align*}
\Big\{&(s_1,\cdots,s_l)\in \mathbb{N}^l\,\,\Big| \,\,\ell(G) = \ell(G -Q_j ) \text{ for all }1\leqslant j \leqslant l, \text{ where }G=\sum_{i=1}^l s_iQ_i \Big\}.
\end{align*}
In addition, they showed in { Lemma 2.5 of \cite{carvalho2005goppa}} that $ (s_1,\cdots,s_l)  $ is a pure gap at $ (Q_1,\cdots, Q_l) $ if and only if
\begin{align*}
\ell(s_1Q_1+\cdots+s_l Q_l)=\ell((s_1-1)Q_1+\cdots+(s_l-1) Q_l).
\end{align*}

A useful way to calculate the Weierstrass semigroups is given as follows, which can be regarded as an easy generalization of {Lemma 2.1 due to Kim \cite{kim}}.
\begin{lemma}[{\cite{carvalho2005goppa}, Lemma 2.2}]
	\label{lem:Weierstsemi}
	For rational places $ Q_1,\cdots,Q_l $ with
	$ 1 \leqslant l \leqslant r $, the set $ H(Q_1,\cdots,Q_l) $ is given by
	\begin{align*}
	\Big\{&(s_1,\cdots,s_l)\in \mathbb{N}_0^l\,\,\Big| \,\,\ell(G) \neq \ell(G -Q_j ) \text{ for all } 1\leqslant j \leqslant l, \text{ where } G=\sum_{i=1}^l s_iQ_i \Big\}.
	\end{align*}
\end{lemma}

In the rest of this section, we will restrict our study to the divisor $  G:=\sum_{\mu=0}^{q-1} r_{\mu}P_{\mu}  +t P_{\infty}  $ and denote $ {\bm{r}}=(r_0,r_1,\cdots,r_{q-1}) $. Our main task is to determine the Weierstrass semigroups and the pure gaps at {distinct rational places} $ P_0, P_1, \cdots, P_{q-1} , P_{\infty} $. Before we proceed, some auxiliary results are presented in the following.  Denote $ \Omega_{\bm{r},\,\bm{0},\,t} $ by $ \Omega_{(r_0,\,r_1,\,\cdots,\,r_{q-1}),\, t }$ for the clarity of description.  

\begin{lemma}\label{lem:omegawith1}
	For the lattice point set $ \Omega_{(r_0,\,r_1,\,\cdots,\,r_{q-1}) ,\, t} $, we have 
	the following assertions.
	\begin{enumerate}
		\item [$ (1) $]
		$ \#	
		\Omega_{(r_0,\,r_1,\,\cdots,\,r_{q-1}) ,\, t }  = \#	\Omega_{(r_0-1,\,r_1,\,\cdots,\,r_{q-1}) ,\, t }+1 $ if and only if 	 
		\begin{align*}  \sum\limits_{\mu=1 }^{q-1} \ceil{ \dfrac{r_0-r_{\mu}}{m(q+1)}}
		+q(q-1) \ceil{ \dfrac{r_0 }{m }}
		\leqslant\frac{ t+ q^3 r_0 }{m(q+1)} .
		\end{align*}  
		\item  [$ (2) $]
		$ \#	
		\Omega_{(r_0,\,r_1,\,\cdots,\,r_{q-1}) ,\, t }  = \#	\Omega_{(r_0,\,r_1,\,\cdots,\,r_{q-1} ), \,t-1 }+1  $ if and only if 	 
		\begin{align*} 
		\sum\limits_{\mu=1 }^{q-1} \ceil{ \dfrac{q^{n-3}t-r_{\mu}}{m(q+1)}}
		+q(q-1) \ceil{ \dfrac{q^{n-3}t }{m }}
		\leqslant t+ \frac{r_0 -q^{n-3}t}{m(q+1)}  .
		\end{align*}
	\end{enumerate}	
	
\end{lemma}

\begin{proof}
	Consider two lattice point sets $\Omega_{(r_0,\,r_1,\,\cdots,\,r_{q-1} ), \,t}$ and $\Omega_{(r_0-1,\,r_1,\,\cdots,r_{q-1} ),\, t}$, which are given in Equation \eqref{eq:omegarceiljmu}. Clearly, the latter one is a subset of the former one, and the complement set $ \Phi $ of $\Omega_{(r_0-1,\,r_1,\,\cdots,\,r_{q-1} ), \,t}$ in $\Omega_{(r_0,\,r_1,\,\cdots,\,r_{q-1} ),\, t}$ is given  by
	\begin{align*}
	\Phi :
	=&\Big\{ (i, \,\bm{j},\,\bm{k}) \,	\Big	|\,\,	i+r_0= 0 ,\\
	& j_{\mu} =  \ceil{\frac{-i-r_{\mu}}{m(q+1)}  }  \text{ for }   \mu =1,\cdots, q-1,\nonumber \\
	& k_{\nu} =  \ceil{\frac{-i}{m}  }  \text{ for }   \nu =1,\cdots, q^2-1,\nonumber \\
	&q^3 i+m(q+1)|{\bm{j}}|+mq|{\bm{k}}|    \leqslant t
	\Big\},
	\end{align*}
	where $ {\bm{j}}=(j_1,j_2,\cdots,j_{q-1}) $ and $ {\bm{k}}=(k_1,k_2,\cdots,k_{q^2-1}) $. 
	It follows immediately that the set $ \Phi $ is not empty if and only if
	\begin{align*} -q^3 r_0 +	m(q+1) \sum\limits_{\mu=1 }^{q-1} \ceil{ \dfrac{r_0-r_{\mu}}{m(q+1)}}
	+mq(q^2-1) \ceil{ \dfrac{r_0 }{m }}
	\leqslant t,
	\end{align*} 
	which concludes the first assertion.

	For the second assertion, we obtain from  Corollary~\ref{cor:thetabase2} that
	the difference between $ \#	\Omega_{(r_0,\,r_1,\,\cdots,\,r_{q-1}) ,\, t }$ and $ \#	 \Omega_{(r_0,\,r_1,\,\cdots,\,r_{q-1} ), \,t-1 } $ is exactly
	the same as the one between $ \#\Theta_{\bm{r},\bm{0},\,t}$ and $\# \Theta_{\bm{r},\,\bm{0},\,t-1}$.
	In an analogous way, we define $ \Psi $ as the
	complementary set of $\Theta_{\bm{r},\,\bm{0},\,t-1}$ in $\Theta_{\bm{r},\,\bm{0},\,t}$, namely
	\begin{align*} 
	\Psi:=\Big\{& (u,\bm{\lambda},\bm{\gamma}) \,\,\Big| \,\,  u  = - t,\nonumber \\
	& 0 \leqslant q^{n-3}u+m \gamma_{\nu} <m \text{ for }  \nu =1,\cdots, q^2-1, \nonumber \\
	& 0 \leqslant q^{n-3}u+m(q+1) \lambda_{\mu}-m |{\bm{\gamma}}|+  r_{\mu} < m(q+1) \text{ for }  \mu =1,\cdots, q-1, \nonumber \\
	& -\Big((m(q+1)-q^{n-3}) u+m(q+1)|{\bm{\lambda}}| + m |{\bm{\gamma}}|\Big) +r_0 \geqslant 0  
	\Big\}.
	\end{align*} 
	The set $ \Psi $ is not empty if and only if
	\begin{align*} 
	m(q+1) \sum\limits_{\mu=1 }^{q-1} \ceil{ \dfrac{q^{n-3}t-r_{\mu}}{m(q+1)}}
	+mq(q^2-1) \ceil{ \dfrac{q^{n-3}t }{m }}
	\leqslant r_0+ (q^n+1-q^{n-3})t  ,
	\end{align*}
	completing the proof of the second assertion.	
\end{proof}

We are now ready for the main results of this section  dealing with Weierstrass semigroups and pure gap sets, which play an interesting role in finding codes with good parameters. For simplicity, we write $ {\bm{r}}_l=(r_0,r_1,\cdots,r_{l}) $ and define
\begin{align*}
W_j({\bm{r}}_l,t,l):=&\sum\limits_{i=0  \atop i\neq j }^{l} \ceil{ \dfrac{r_j-r_i}{m(q+1)}}
+(q-1-l)\ceil{ \dfrac{r_j }{m(q+1) }}   +q(q -1) \ceil{ \dfrac{r_j }{m }}
- \dfrac{t+q^3 r_j }{m(q+1)},  
\intertext{for $ j =1,\cdots,l $,  and } 
W_{\infty}({\bm{r}}_l,t,l):=&\sum\limits_{i=1 }^{l} \ceil{ \dfrac{q^{n-3}t-r_{i}}{m(q+1)}}
+(q-1-l)\ceil{ \dfrac{q^{n-3}t}{m(q+1)}}\\
&+q(q-1) \ceil{ \dfrac{q^{n-3}t }{m }}
- t- \frac{r_0-q^{n-3}t }{m(q+1)}.
\end{align*} 
\begin{theorem}
	\label{lem:Weierstsemi2}
	Let $ P_0,P_1,\cdots,P_l $ be rational places as defined previously. For $  0 \leqslant l < q $, the following assertions hold.
	\begin{enumerate}
		\item [$ (1) $] The Weierstrass semigroup $ H(P_0,P_1,\cdots,P_l)$ is given by
		\begin{align*}
		\Big \{ (r_0,r_1,\cdots,r_{l})\in \mathbb{N}_0^{l+1}  \,\,\Big| \,\, W_j({\bm{r}}_l,0,l)
		\leqslant 0 \text{ for all }   0 \leqslant j \leqslant l
		\Big \}.
		\end{align*} 
		\item [$ (2) $] The Weierstrass semigroup $ H(P_0,P_1,\cdots,P_l,P_{\infty})$ is given by
		\begin{align*}
		\Big \{ (r_0,r_1,\cdots,r_{l},t)\in \mathbb{N}_0^{l+2}  \,\,\Big| \,\,
		W_j({\bm{r}}_l,t,l)
		\leqslant 0 \text{ for all }   0 \leqslant j \leqslant l \text{ and }
		{	W_{\infty}({\bm{r}}_l,t,l)
			\leqslant 0   }
		\Big \}.
		\end{align*}		
		\item [$ (3) $]
		The pure gap set $ G_0(P_0,P_1,\cdots,P_l)$ is given by
		\begin{align*}
		\Big\{(r_0,r_1,\cdots,r_l)\in \mathbb{N}^{l+1}\,\,\Big |\,\,
		W_j({\bm{r}}_l,0,l)
		> 0 \text{ for all }   0 \leqslant j \leqslant l \Big\}.
		\end{align*}
		\item [$ (4) $] The pure gap set $ G_0(P_0,P_1,\cdots,P_l,P_{\infty})$ is given by
		\begin{align*}
		\Big \{ (r_0,r_1,\cdots,r_{l},t)\in \mathbb{N}^{l+2}  \,\,\Big| \,\,
		W_j({\bm{r}}_l,t,l)
		> 0 \text{ for all }   0 \leqslant j \leqslant l \text{ and } 	{ W_{\infty}({\bm{r}}_l,t,l)
		> 0 }
		\Big \}.
		\end{align*}
	\end{enumerate}
\end{theorem}
\begin{proof}
	The desired conclusions follow from Theorem \ref{thm:basis1},  Lemmas \ref{lem:omeganoorder}, \ref{lem:Weierstsemi} and \ref{lem:omegawith1}.	
\end{proof}

{In the literature, Beelen and Montanucci \cite[Theorem 1.1]{Beelen2018} studied the Weierstrass semigroups $ H(P) $ at any point $ P $ on the GK curve. For the GGS curve, $ H(P_0) $ was
determined in \cite[Proposition 4.3]{Bartoli2017} and $ G(P_0) $ was described independently in \cite[Corollary 18]{Barelli2018} for $ P_0 \neq P_{\infty} $, while $ H(P_{\infty}) $ and $ G(P_{\infty}) $ were given in \cite[Corollary 3.5 and Theorem 3.7]{Guneri2013}, respectively. Here we give another characterization and an alternative proof of $ H(P_0) $ and $ G(P_0) $ using our main theorem.}
\begin{corollary} 
	With notation as before, we have the following statements.\\
	 $ (1) $ 
		$ H(P_0)= \Big \{ k\in \mathbb{N}_0   \,\,\Big| \,\, 
		(q-1)\ceil{ \dfrac{k }{m(q+1) }} 
		+q(q -1) \ceil{ \dfrac{k }{m }}
		\leqslant \dfrac{q^3 k}{m(q+1)}  
		\Big \} $.\\
		 $ (2) $ Let $ a,b,c\in \mathbb{Z} $. 
		Then $ a+m(b+(q+1)c)\in \mathbb{N} $ is a gap at 
		$P_0$ if and only if exactly one of the following two conditions is satisfied:
		\begin{enumerate}
			\item [$ (i) $]
			 $ a=0 $, $ 0<b\leqslant q-1 $ and $ 0 \leqslant c \leqslant q-1-b $;
		
		\item [$ (ii) $] $ 0< a<m $, $ 0\leqslant b\leqslant q $, $ 0 \leqslant c \leqslant q^2-1-b+ \floor{\dfrac{b}{q+1}-\dfrac{q^3 a}{m(q+1)}} $ and  \\ $ b- \floor{\dfrac{b}{q+1}-\dfrac{q^3 a}{m(q+1)}} \leqslant q^2-1 $.

	\end{enumerate}
\end{corollary}
\begin{proof}
	The first statement is an immediate consequence of Theorem \ref{lem:Weierstsemi2} $ (1) $.
	
	We now focus on the second statement. It follows from Theorem \ref{lem:Weierstsemi2} $ (3) $ that the Weierstrass gap set at $ P_0 $ is
	\begin{align*}
	G(P_0)=
	\Big \{ k\in \mathbb{N}   \,\,\Big| \,\, 
	(q-1)\ceil{ \dfrac{k }{m(q+1) }} 
	+q(q -1) \ceil{ \dfrac{k }{m }}
	> \dfrac{q^3 k}{m(q+1)}  
	\Big \} .
	\end{align*}
	Let $ k \in G(P_0) $ and write $ k=a+m(b+(q+1)c) $, where 
	$ 0\leqslant a<m $, $ 0 \leqslant b \leqslant q$ and $ c \geqslant 0  $. 
	We find that the case $ a+mb=0 $ does not occur, since otherwise 
	we have $ k=m(q+1)c $ and 
	\begin{align*}
	&(q-1)\ceil{ \dfrac{k }{m(q+1) }} 
	+q(q -1) \ceil{ \dfrac{k }{m }}	- \dfrac{q^3 k}{m(q+1)}  \\
	& = (q-1)c +q(q -1)(q+1)c - q^3 c \\
	& = -c \leqslant 0,
	\end{align*}
	which contradicts to the fact $ k \in G(P_0) $. So $ a+mb \neq 0$. There are two possibilities.
	
	$ (i) $ If $ a=0$, then $ 0<b\leqslant q $. In this case, $ k= mb+m(q+1)c $
	is a gap at $ P_0 $ if and only if 
	\begin{align*}
	(q-1)\ceil{ \dfrac{k }{m(q+1) }} 
	+q(q -1) \ceil{ \dfrac{k }{m }}	> \dfrac{q^3 k}{m(q+1)},
	\end{align*}
	or equivalently,
	\begin{align*}
	(q-1)(c+1) 
	+q(q -1) (b+(q+1)c)	> q^3c+\dfrac{q^3 b}{q+1},
	\end{align*}
	leading to the first condition $ 0 \leqslant c \leqslant q-1-b $ and $ 0<b\leqslant q-1 $. 
	
	$  (ii)  $ If $ 0< a<m $,  then $ 0\leqslant b\leqslant q $. In this case, we have similarly that $ k= a+mb+m(q+1)c $
	is a gap at $ P_0 $ if and only if 
	\begin{align*}
	(q-1)(c+1) 
	+q(q -1) (1+b+(q+1)c)	> q^3c+\dfrac{q^3 a}{m(q+1)}+\dfrac{q^3 b}{q+1},
	\end{align*}
	which gives the second condition
	$  0 \leqslant c \leqslant q^2-1-b+ \floor{\dfrac{b}{q+1}-\dfrac{q^3 a}{m(q+1)}} $. 
	Note that
	\begin{align*}
	q^2-1-b+ \floor{\dfrac{b}{q+1}-\dfrac{q^3 a}{m(q+1)}} \geqslant q^2-1  -q -1 \geqslant 0.
	\end{align*}
	The proof is finished. 	
\end{proof}

\section{The floor of divisors }\label{sec:floor}

In this section, we investigate the floor of divisors from GGS curves. The
significance of this concept is that it provides a useful
tool for evaluating parameters of AG codes. We begin with
general function fields.
\begin{definition}[{\cite{Maharaj2005riemann}, Definition 2.2}]
	Given a divisor $ G $ of a function field $ F/\mathbb{F}_q $ with $ \ell(G)>0 $, the floor of $ G $ is the unique divisor $ G' $ of $ F $ of minimum degree such that $ \mathcal{L}(G)=\mathcal{L}(G') $. The floor of $ G $ will be denoted by $ \lfloor G \rfloor $.
\end{definition}

The floor of a divisor can be used to characterize Weierstrass semigroups and pure gap sets. Let $ G=s_1Q_1 + \cdots+ s_lQ_l $. It is not hard to see that $ (s_1,\cdots,s_l) \in H (Q_1,\cdots,Q_l)$ if and only if $ \lfloor G \rfloor=G $. Moreover,  $ (s_1,\cdots,s_l) $ is a pure gap at $ (Q_1,\cdots,Q_l) $ if and only if
\begin{equation*}
\lfloor G \rfloor= \lfloor (s_1-1)Q_1 + \cdots+ (s_l-1)Q_l \rfloor.
\end{equation*}

Maharaj, Matthews and Pirsic in \cite{Maharaj2005riemann} defined the floor of a divisor and characterized it by the basis of the Riemann-Roch space.
\begin{theorem}[{\cite{Maharaj2005riemann}, Theorem 2.6}]\label{thm:floorofG}
	Let $ G $ be a divisor of a function field
	$ F/\mathbb{F}_q $ and let $ b_1, \cdots, b_t \in  \mathcal{L}(G)$ be a spanning set for $ \mathcal{L}(G)$. Then
	\begin{equation*}
	\lfloor G \rfloor=-\gcd\Big\{\Div(b_i)\,\,\Big| \,\,i=1,\cdots,t\Big\}.
	\end{equation*}
\end{theorem}

The next theorem extends Theorem 3.4 of \cite{carvalho2005goppa} by determining the lower bound of minimum distance in a more general situation.

\begin{theorem}[{\cite{Maharaj2005riemann}, Theorem 2.10}]\label{thm:Maharajfloor}
	Assume that $ F/\mathbb{F}_q $ is a function field with genus $ g $.
	Let $ D:=Q_1+\cdots+Q_N $ where $ Q_1,\cdots, Q_N $ are distinct rational places of $ F $, and let $ G:= H+\lfloor H \rfloor $ be a divisor of $ F $
	such that $ H $ is an effective divisor whose support does not contain any of the places $ Q_1,\cdots, Q_N $. Then the minimum distance of $ C_{\Omega} (D,G)$ satisfies
	\begin{align*}
	d_{\Omega} \geqslant 2 \deg(H)-(2g-2).
	\end{align*}
\end{theorem}

The following theorem provides a characterization of the floor over GGS curves, which can be viewed as a generalization of Theorem 3.9 in \cite{Maharaj2005riemann} related to Hermitian function fields.
\begin{theorem}\label{thm:floorofH}
	Let $ H :=\sum_{\mu=0}^{q-1} r_{\mu}P_{\mu} + \sum_{\nu=1}^{q^2-1}s_{\nu} Q_{\nu}+tP_{\infty} $ be a divisor of the GGS curve given by \eqref{eq:GGScurve}. Then the floor of $ H $ is given by
	\begin{align*}
	\lfloor H \rfloor =\sum_{\mu=0}^{q-1} r_{\mu}'P_{\mu} + \sum_{\nu=1}^{q^2-1}s_{\nu}' Q_{\nu}+t'P_{\infty},
	\end{align*}
	where
	\begin{align*}
	r_0'&=\max\Big\{-i \,\,\Big| \,\, (i,\bm{j},\,\bm{k}) \in \Omega_{\bm{r},\,\bm{s},\,t}  \Big\},\\
	r_{\mu}'&= \max\Big\{ -i-m(q+1) j_{\mu}\,\,\Big| \,\,(i,\bm{j},\,\bm{k}) \in \Omega_{\bm{r},\,\bm{s},\,t}  \Big\}
	\text{ for } \mu = 1,\cdots, q-1,\\
	s_{\nu}'&= \max\Big\{ -i-m k_{\nu}\,\,\Big| \,\,(i,\bm{j},\,\bm{k}) \in \Omega_{\bm{r},\,\bm{s},\,t}  \Big\}
	\text{ for } \nu = 1,\cdots, q^2-1,\\
	t'& = \max\Big\{q^3 i+m(q+1)|{\bm{j}}|+mq|{\bm{k}}|\,\,\Big| \,\,(i,\bm{j},\,\bm{k}) \in \Omega_{\bm{r},\,\bm{s},\,t}  \Big\}.
	\end{align*}
\end{theorem}
\begin{proof}
	Let $ H =\sum_{\mu=0}^{q-1} r_{\mu}P_{\mu} + \sum_{\nu=1}^{q^2-1}s_{\nu} Q_{\nu}+tP_{\infty} $. It follows from Theorem \ref{thm:basis1} that the elements $ E_{i,\,\bm{j},\,\bm{k}} $ of Equation \eqref{eq:Eijk} with $ (i,\bm{j},\,\bm{k}) \in \Omega_{\bm{r},\,\bm{s},\,t}  $ form a basis for the Riemann-Roch space
	$ \mathcal{L}(H) $.
	Note that the divisor of $ E_{i,\,\bm{j},\,\bm{k}} $  is
 {	\begin{align*}
	& iP_0 + \sum_{\mu=1}^{q-1}(i+m(q+1) j_{\mu}) P_{\mu}
	+\sum_{\nu=1}^{q^2-1}(i+m k_{\nu})Q_{\nu} \nonumber \\
	& -\big(q^3 i+m(q+1)|{\bm{j}}|+mq|{\bm{k}}| \big) P_{\infty} .
	\end{align*}
}	
	By Theorem \ref{thm:floorofG}, we get that
	\begin{equation*}
	\lfloor H \rfloor =-\gcd\Big\{\Div(E_{i,\,\bm{j},\,\bm{k}})\,\,\Big| \,\, (i,\bm{j},\,\bm{k}) \in \Omega_{\bm{r},\,\bm{s},\,t}  \Big\}.
	\end{equation*}	
	The desired conclusion then follows. 
\end{proof}

\section{Examples of codes on GGS curves}\label{sec:Examples}

In this section we treat several examples of codes to illustrate our results. The codes in the next example will give new records of better parameters than the corresponding ones in MinT's tables \cite{MinT}.

\begin{example}
	Now we study codes arising from GK curves, that is, we let $ q=2 $ and $ n=3 $ in \eqref{eq:GGScurve} given at the beginning of Section \ref{sec:Bases}. This curve has $ 225 $ $ \mathbb{F}_{64} $-rational points and its genus is $ g=10 $.  Here we will study multi-point AG codes from this curve by employing our previous results.  
	
	Let us take $ H=3P_0+4P_1+11P_{\infty} $ for example. Then	 
	it can be computed from Equation \eqref{eq:omegarceiljmu} that the elements $ \big(-i,\,-i-m(q+1)j_1,\,-i-m k_1,\,-i-mk_2,\,-i-mk_3,\,q^3i+m(q+1)j_1+mq(k_1+k_2+k_3)\big) $, with $(i,\,j_1,\,k_1,\,k_2,\,k_3)  \in  \Omega_{3,\,4,\,0,\,0,\,0,\,11} $, are as follows:
	\begin{align*}
	&( \phantom{-}3, \phantom{-}3, \phantom{-}0, \phantom{-}0, \phantom{-}0, -6\phantom{1}), \\
	&( \phantom{-}2, \phantom{-}2, -1, -1, -1, \phantom{1}2\phantom{-}), \\
	&( \phantom{-}1, \phantom{-}1, -2, -2, -2, 10\phantom{-}), \\
	&( \phantom{-}0, \phantom{-}0, \phantom{-}0, \phantom{-}0, \phantom{-}0, \phantom{1}0\phantom{-}), \\
	& (-1, -1, -1, -1, -1, \phantom{1}8\phantom{-}), \\
	&(-3, -3, \phantom{-}0, \phantom{-}0, \phantom{-}0, \phantom{1}6\phantom{-}), \\
	&(-6, \phantom{-}3, \phantom{-}0, \phantom{-}0, \phantom{-}0, \phantom{1}3\phantom{-}), \\
	&(-7, \phantom{-}2, -1, -1, -1, 11\phantom{-}), \\
	&(-9, \phantom{-}0, \phantom{-}0, \phantom{-}0, \phantom{-}0, \phantom{1}9\phantom{-}). 
	\end{align*}	 
	So we obtain from Theorem \ref{thm:floorofH} that $ \lfloor H \rfloor = 3P_0+3P_1 + 11 P_{\infty} $. Let $ D $ be a divisor consisting of $ N=216 $ rational places away from the places $ P_0,P_1$, $ Q_1 $, $ Q_2 $, $ Q_3 $ and $P_{\infty} $. According to Theorem \ref{thm:Maharajfloor}, if we let
	$ G=H+\floor{H}  = 6 P_0 +7 P_1 + 22 P_{\infty} $, then the code
	$ C_{\Omega}(D,G) $ has minimum distance at least $  2 \deg(H)-(2g-2)=18 $. 
	Since $ 2g-2<\deg(G)<  N $, the dimension of $ C_{\Omega} (D,G) $ is $ k_{\Omega}= N+g-1-\deg(G)= 190 $. In other words, the code $ C_{\Omega}(D,G) $ has parameters $ [216,\,190,\,\geqslant 18] $. One can verify that our resulting code improve the minimum distance with respect to MinT's Tables. Moreover $ C_{\Omega}(D,G) $ is equivalent to $ C_{\mathcal{L}}(D,G') $, where $ G'=2 P_0 +P_1 + 8 Q_1 + 8 Q_2 + 8 Q_3+ 4 P_{\infty}  $, and
	its generating matrix can be determined by Theorem \ref{th:dualcode}.
	
	Additionally, we remark that more AG codes with excellent parameters can be found by taking $ H= a P_0 + b P_1 + 7 P_{\infty}$, where $ a,b\in \{4,\,5,\,6\} $ and $ 9 \leqslant a+b \leqslant 12 $. The floor of such $ H $ is computed to be $ \floor{H}= a P_0 + b P_1 + 6 P_{\infty} $. Let $ D $ be as before. If we take $ G =H+\floor{H}  = 2a P_0 +2b P_1 + 13 P_{\infty} $, then we can produce AG codes  $ C_{\Omega}(D,G) $
	with parameters $ [216,\, 212-2a-2b, \,\geqslant 2a+2b-4] $. All of these codes
	improve the records of the corresponding ones found in MinT's Tables. 
\end{example}
\begin{example}
	Consider the curve $ \textup{GGS}(q,n) $ of \eqref{eq:GGScurve} with $ q=2 $ and $ n=5 $. This curve has $ 3969 $ $ \mathbb{F}_{2^{10}} $-rational points and its genus is $ g=46 $. It follows from Theorem \ref{lem:Weierstsemi2}  that
	\begin{align*}
	\Big\{(57,\,j,\,3)\,\,\Big| \,\,  1 \leqslant j \leqslant 3\Big\} \subseteq G_0(P_0,P_1,P_{\infty}).
	\end{align*}
	Let $ D $ be a divisor consisting of $ N=3960 $ rational places except $ P_0$, $P_1$, $Q_1$, $Q_2$, $Q_3$ and $P_{\infty} $. Applying Theorem 3.4 of \cite{carvalho2005goppa} (see also Theorem 1, \cite{HuYang2016Kummer}), if we take $ G=113 P_0+3P_1+5P_{\infty} $, then the three-point code $ C_{\Omega}(D,G) $ has length $  N=3960 $, dimension $   N+g-1-\deg (G)=3884 $ and minimum distance  at least
	$ 36 $. Thus we produce an AG code with parameters $ [ 3960, \,3884,\,\geqslant 36] $. Unfortunately, this $ \mathbb{F}_{2^{10}} $-code cannot be compared with the one in MinT's Tables because the alphabet size given is at most $ 256 $. 
\end{example}

\section*{Acknowledgements}
The authors are very grateful to the editor and the anonymous reviewers for their valuable comments and suggestions that improved the quality of this paper.

\medskip
Received xxxx 20xx; revised xxxx 20xx.
\medskip

\end{document}